\documentclass[11pt,reqno]{amsart}

\usepackage{amscd,amssymb,amsmath,amsthm}
\usepackage{graphicx}
\usepackage{color}
\usepackage{cite}
\newtheorem{thm}{Theorem}
\newtheorem{defn}{Definition}
\newtheorem{lemma}{Lemma}

\newtheorem{rk}{Remark}

\numberwithin{equation}{section} 

\makeatletter
\newcommand*{\rom}[1]{\expandafter\@slowromancap\romannumeral #1@}
\makeatother

\begin{document}
	\title [Fixed points of an operator]
	{Fixed points of an infinite dimensional operator related to  Gibbs measures}
	
	\author {U.R. Olimov, U.A. Rozikov}
	
	\address{U.\ R. Olimov\\ V.I.Romanovskiy Institute of Mathematics,  9, Universitet str., 100174, Tashkent, Uzbekistan.}
	\email {umrbek.olimov.92@mail.ru}
	
	\address{ U.Rozikov$^{a,b,c}$\begin{itemize}
			\item[$^a$] V.I.Romanovskiy Institute of Mathematics,  9, Universitet str., 100174, Tashkent, Uzbekistan;
			\item[$^b$] AKFA University,
			264, Milliy Bog street,	Yangiobod QFY, Barkamol MFY,
			Kibray district, 111221, Tashkent region, Uzbekistan;
			\item[$^c$] National University of Uzbekistan,  4, Universitet str., 100174, Tashkent, Uzbekistan.
	\end{itemize}}
	\email{rozikovu@yandex.ru}
	
	\	\begin{abstract} We describe fixed points of an infinite dimensional non-linear operator related to a hard core (HC) model with a countable set  $\mathbb{N}$ of spin values on the Cayley tree. This operator is defined by a countable set of parameters $\lambda_{i}>0$, $a_{ij}\in \{0,1\}$, $ i,j \in \mathbb{N}$. We find a sufficient condition on these parameters under which the operator has unique fixed point. When this condition is not satisfied then we show that the operator may have up to five fixed points.  Also, we prove that every fixed point generates a normalisable boundary law and therefore defines a Gibbs measure for the given HC-model.
	\end{abstract}
	\maketitle
	
	{\bf Mathematics Subject Classifications (2010).} 82B05, 82B20
	(primary); 60K35 (secondary)
	
	{\bf{Key words.}} Fixed point, Cayley tree, Gibbs measure, HC model.
	
	\section{Introduction}
In this paper we investigate an infinite dimensional operator $F$ related to a physical system with  $\mathbb N$-valued spin variables $\sigma(x)$ located on the vertices $x$ of a ($k$-regular) Cayley tree, where each vertex has $k+ 1$ neighbors. We are interested to  fixed points of the operator $F$. 

It is known (see \cite{Z1}) that each normalisable fixed point of operator $F$ defines a Gibbs measure of the  $\mathbb N$-valued spin system. A non-normalisable fixed point does not define a Gibbs measure, but if the operator corresponds to a gradient potential on the space of gradient configurations $\sigma$, then such fixed points define some gradient Gibbs measures (for detailed motivation and recent results see \cite{BiKo}, \cite{BEvE}, \cite{FV}--\cite{K}, \cite{Ra}, \cite{Sh}, \cite{Ve}).

Theory of Gibbs measures on trees have been mostly developed for Hamiltonians with finite set of spin values (for example, the Ising model, the Potts models, and  hard-core models). For such models the translation invariant Gibbs measures can be described in terms of the roots
of polynomials depending on parameters of the model and on the order $k$ of the Cayley tree (see \cite{BR}, \cite{KRK}, \cite{Ro}, \cite{Roz} and references therein). 

In the case of $\mathbb{N}$-valued spins the investigation of (gradient) Gibbs measures is more difficult as the solutions of a corresponding equation became an infinite-dimensional vector, even for the translation invariant measures, so we can not hope for explicit solutions in the general case. There may be no solutions at all, due to non-compactness of the set $\mathbb{N}$.

 In this paper we consider models for which such solutions and corresponding Gibbs measures do exist. Moreover, we show that depending on  parameters there may be up to five translation invariant Gibbs measures.

\section{Condition of uniqueness of fixed point. }
 Denote
$$l^{1}_{+}=\left\{x=(x_{1},x_{2},\dots,x_{n},\dots) \, : \, x_{i}>0, \, \|x\|=\sum_{j=1}^{\infty}x_j<\infty\right\}.$$

To describe translation invariant Gibbs measures of hard-core (HC) models on a Cayley tree of order $k\geq 2$ one has to study fixed points of operator $F: l^{1}_{+}\to l^{1}_{+}$ defined by
$$F: \ \ x'_{i}=\lambda_{i}\left(\frac{1+\sum_{j=1}^{\infty}a_{ij}x_j}{1+\sum_{j=1}^{\infty}a_{1j}x_j}\right)^k,$$
where $k, i\in \mathbb{N}$,  $\lambda_{i}>0$, and  $a_{ij}\in{\{0,1\}}$ are given parameters.

 In this paper we are going to study fixed points of $F$, in the case when $a_{1j}=1$ for any $j\in \mathbb N$. In this case the operator takes a simpler form:

\begin{equation}\label{1}
	F: \ \ x'_{i}=\lambda_{i}\left(\frac{1+\sum_{j=1}^{\infty}a_{ij}x_j}{1+\|x\|}\right)^k,
\end{equation}
where $k, i\in \mathbb{N}$,  $\lambda_{i}>0$.

Let $\lambda=(\lambda_1,\lambda_2,\dots)\in l^{1}_{+}$.
Denote
$$A_{k}=\left\{x\in l^1_+:\lambda_1+\frac{\|\lambda\|-\lambda_1}
{(1+\|\lambda\|)^k}\leq \|x\|\leq \|\lambda \|\right\}.$$
\begin{lemma}\label{lemma1}
If $\lambda=(\lambda_1,\lambda_2,\dots)\in l^{1}_{+}$ then $A_{k}$ is an invariant with respect to operator (\ref{1}), $F:l^1_+\to l^1_+$, i.e.,  $F(A_{k})\subset A_{k}$.
\end{lemma}
\begin{proof} For any  $x\in \ell^1_{+}$ from (\ref{1}) we get
$$0<x'_{i}=\lambda_{i}\left(\frac{1+\sum_{j=1}^{\infty}a_{ij}x_j}{1+\|x\|}\right)^k\leq\lambda_{i}\left(\frac{1+\|x\|}{1+\|x\|}\right)^k=\lambda_{i}.$$
Consequently,
$$\|x'\|=\sum_{i=1}^{\infty}x'_i\leq\sum_{i=1}^{\infty}\lambda_i=\|\lambda\|.$$
Now we find a lower bound for $\|x'\|$. For $x\in A_k$, we have from (\ref{1}) that
$$x'_{1}=\lambda_{1},$$
$$x'_{i}=\lambda_{i}\left(\frac{1+\sum_{j=1}^{\infty}a_{ij}x_j}{1+\|x\|}\right)^k\geq\frac{\lambda_{i}}{(1+\|x\|)^k}\geq\frac{\lambda_{i}}{(1+\|\lambda\|)^k}, \ \ i\geq 2.$$
Therefore,
 $$\|x'\|\geq\lambda_{1}+\sum_{j=2}^{\infty}\frac{\lambda_{i}}{(1+\|\lambda\|)^k}=\lambda_1+\frac{\|\lambda\|-\lambda_1}{(1+\|\lambda\|)^k}.$$
\end{proof}
Denote
$$\hat{\lambda}:=\hat{\lambda}(k) = \frac{1}{6k}\left(1-3k+\sqrt[3]{1+9k+72k^2+\sqrt{(1+9k+72k^2)^2 -(3+18k+9k^2)^3}}\right.$$
$$\left.+\sqrt[3]{1+9k+72k^2-\sqrt{(1+9k+72k^2)^2-(3+18k+9k^2)^3}}\right).$$
By a plot one can see that $\hat\lambda(k)>0$ is a decreasing function of $k\geq 2$, with maximal value $\hat\lambda(2)\approx 0.5296$.
\begin{lemma}\label{lemma2} For any  $\lambda\in\ell^1_{+}$ with $\|\lambda\|<\hat{\lambda}$, there exists $\kappa(\lambda,k)\in (0, 1)$ such that
$$\|F(x)-F(y)\|<\kappa\|x-y\|,  \ \ \mbox{for all} \ \ x, y\in A_{k}.$$
Thus $F:A_k\to A_k$ is a contraction.
\end{lemma}
\begin{proof} Recall that $x_1'=\lambda_1$. Take any  $(\lambda_{1},x_{2},x_{3},\dots,x_{n},\dots)$ and $(\lambda_{1},y_{2},y_{3},\dots,y_{n},\dots)$.
Then we have
$$F_{i}(x)-F_{i}(y)=\lambda_{i}\left(\left(\frac{1+\sum_{j=1}^{\infty}a_{ij}x_j}{1+\|x\|}\right)^k-\left(\frac{1+\sum_{j=1}^{\infty}a_{ij}y_j}{1+\|y\|}\right)^k\right)$$
$$=\lambda_{i}\left(\frac{1+\sum_{j=1}^{\infty}a_{ij}x_j}{1+\|x\|}-\frac{1+\sum_{j=1}^{\infty}a_{ij}y_j}{1+\|y\|}\right)U(x,y),
$$where
$$ U(x,y):=
\sum_{p=0}^{k-1}\left(\frac{1+\sum_{j=1}^{\infty}a_{ij}x_j}{1+\|x\|}\right)^{k-p-1}\cdot\left(\frac{1+\sum_{j=1}^{\infty}a_{ij}y_j}{1+\|y\|}\right)^{p}.$$

Consequently,
$$F_{i}(x)-F_{i}(y)= U(x,y)\lambda_{i} \frac{(1+\sum_{j=1}^{\infty}a_{ij}x_j)(1+\|y\|)-(1+\sum_{j=1}^{\infty}a_{ij}y_j)(1+\|x\|)}{(1+\|x\|)(1+\|y\|)}$$
$$=U(x,y)\lambda_{i}\frac{(1+\|x\|)\sum_{j=1}^{\infty}a_{ij}(x_j-y_j)+(\|y\|-\|x\|)\left(1+\sum_{j=1}^{\infty}a_{ij}x_j\right)}{(1+\|x\|)(1+\|y\|)}.$$

Now we use the following inequalities:

1) If $x,y\in\ell^1_{+}$ then $\|y\|-\|x\|=\sum_{j=1}^{\infty}(y_{j}-x_{j})$ and $\|y-x\|=\sum_{j=1}^{\infty}|y_{j}-x_{j}|$, therefore, $$-\|y-x\|\leq\|y\|-\|x\|\leq\|y-x\|.$$

2) For any $x,y\in\ell^1_{+}$, by $a_{ij}\in{\{0,1\}}$ we get
$$\sum_{j=1}^{\infty}a_{ij}x_{j}\leq\|x\|, \ \ \sum_{j=1}^{\infty}a_{ij}(x_{j}-y_{j})\leq\|x-y\|.$$


3) We note also that
$$\|y\|-\|x\|+\sum_{j=1}^{\infty}a_{ij}(x_j-y_j)$$
$$=\sum_{j=1}^{\infty}(y_j-x_j)-\sum_{j=1}^{\infty}a_{ij}(y_j-x_j)=-\sum_{j=1}^{\infty}(1-a_{ij})(x_j-y_j).$$
Thus we have $$-\|x-y\|\leq\sum_{j=1}^{\infty}(1-a_{ij})(x_j-y_j)\leq\|x-y\|.$$

Since $a_{ij}\in\{0, 1\}$ and $x_i>0$ we have $U(x,y)\leq k$ for any $x,y$.
Using the above-mentioned inequalities we obtain:
$$|F_{i}(x)-F_{i}(y)|\leq
\frac{k\lambda_{i}(1+2\|x\|)}{(1+\|x\|)(1+\|y\|)}\|x-y\|.$$
$$\|F(x)-F(y)\|=\sum_{j=1}^{\infty}|F_{i}(x)-F_{i}(y)|
\leq\frac{k\|\lambda\|(1+2\|x\|)}{(1+\|x\|)(1+\|y\|)}\|x-y\|.$$

Denote
$$K(x,y)=\frac{k\|\lambda\|(1+2\|x\|)}{(1+\|x\|)(1+\|y\|)}=\frac{k\|\lambda\|(1+2\|x\|)}{1+\|x\|}\frac{1}{1+\|y\|}.$$
To find upper bound of $K(x,y)$, we introduce  $$g(t)=\frac{1+2t}{1+t}, \ \ h(t)=\frac{1}{1+t}, \ \ t>0.$$
We find maximal values of these functions for $t$ which satisfies:
$$\lambda_1+\frac{\|\lambda\|-\lambda_1}{(1+\|\lambda\|)^2}\leq t \leq\|\lambda\|.$$

It is clear that  $g(t)$ is an increasing function with maximal value  $$g_{\max}(\|\lambda\|)=\frac{1+2\|\lambda\|}{1+\|\lambda\|}.$$
Moreover, function $h(t)$ is a decreasing function and $$h_{\max}(\lambda_1+\frac{\|\lambda\|-\lambda_1}{(1+\|\lambda\|)^2})=\frac{1}{1+\lambda_1+\frac{\|\lambda\|-\lambda_1}{(1+\|\lambda\|)^2}}<\frac{1}{1+\frac{\|\lambda\|}{(1+\|\lambda\|)^2}}.$$
Using these values we obtain
$$K(x,y)<k\|\lambda\|\frac{(1+2\|\lambda\|)}{(1+\|\lambda\|)}\frac{1}{1+\frac{\|\lambda\|}{(1+\|\lambda\|)^2}}=k\frac{2\|\lambda\|^3+3\|\lambda\|^2+\|\lambda\|}{\|\lambda\|^2+2\|\lambda\|+2}.$$
Now we want to find $\lambda$ such that
$$\frac{2k\|\lambda\|^3+3k\|\lambda\|^2+k\|\lambda\|}{\|\lambda\|^2+2\|\lambda\|+2}<1.$$
That is
$$\varphi(\|\lambda\|):=2k\|\lambda\|^3+(3k-1)\|\lambda\|^2+(k-2)\|\lambda\|-2<0.$$
Note that $\varphi(0)=-2<0$ and $\varphi(2)>0$. Therefore $\varphi$ has at least one zero in $(0,2)$. According to Descartes' theorem, $\varphi(\|\lambda\|)=0$ has unique positive solution, since the signs of its coefficients change only once.
By Cardano formula we obtain explicit form of the unique root, which is $\hat \lambda$ defined above.

%
%
%
%
%
%
%

Thus $\varphi(\|\lambda\|)<0$ if $\|\lambda\|<\hat\lambda$.

\end{proof}
For a contraction mapping the following theorem is known:
\begin{thm} If $\|\lambda\|<\hat{\lambda}$ then the operator (\ref{1}) has unique fixed point $z^{*}$ and for any initial point $z^{(0)}\in A_{k}$ we have $\lim\limits_{n \to \infty}F^n(z^{(0)})=z^{*}$.
\end{thm}
\section{Examples of uniqueness}
In this section we give some examples of operator $F$, which has unique fixed point.\\

1. If for all $i,j$ we assume  $a_{ij}=1$ (or $a_{ij}=0$) then operator can be written as
$$x'_{i}=\lambda_{i}\left(\frac{1+\sum_{j=1}^{\infty}a_{ij}x_j}{1+\sum_{j=1}^{\infty}a_{1j}x_j}\right)^k=\lambda_i, \ \ \forall i\in \mathbb N.$$

Thus $F(A_k)={\lambda}$, i.e., any point $x\in A_k$ after first action of $F$, goes to $\lambda$.\\

2. Let $k=2$. If for all $i,j$ we have $a_{1j}=1$, $a_{i1}=1$ and remaining  $a_{ij}=0$ then the operator becomes
\begin{equation}\label{e2o} x'_{1}=\lambda_{1}, \ \ \
x'_{i}=\lambda_{i}\left(\frac{1+\lambda_{1}}{1+\|x\|}\right)^2, i\geq 2.
\end{equation}
To find fixed points of this operator we have to solve
\begin{equation}\label{e2} x_{1}=\lambda_{1}, \ \ \
	x_{i}=\lambda_{i}\left(\frac{1+\lambda_{1}}{1+\|x\|}\right)^2, i\geq 2.
\end{equation}
Summing all equations of this system we get
$$\|x\|=\lambda_1+\frac{(1+\lambda_1)^2}{(1+\|x\|)^2}(\|\lambda\|-\lambda_1).$$
This is
$$\psi(\|x\|):=\|x\|^3+(2-\lambda_1)\|x\|^2+(1-2\lambda_1)\|x\|-(1+\lambda_1)^2(\|\lambda\|-\lambda_1)=0.$$

Note that $\psi(0)<0$ and $\psi(+\infty)>0$. Therefore $\psi$ has at least one root in $(0,+\infty)$. According to Descartes' theorem, $\psi(\|x\|)=0$ has unique positive root, since the signs of its coefficients change only once for each fixed $\lambda_1$.
By Cardano formula we obtain explicit form of the unique root:

%
%
%
%
%
%
%
%
%
%

$$\|x\|=\frac{1+\lambda_1}{3}-1+\alpha+\beta,$$
where
$$\alpha=\sqrt[3]{\frac{(1+\lambda_1)^2}{6}}\cdot\sqrt[3]{1+\sqrt{(9\|\lambda\|+\frac{2-25\lambda_1}{3})^2-\frac{4(1+\lambda_1)^2}{3}}},$$
$$\beta=\sqrt[3]{\frac{(1+\lambda_1)^2}{6}}\cdot\sqrt[3]{1-\sqrt{(9\|\lambda\|+\frac{2-25\lambda_1}{3})^2-\frac{4(1+\lambda_1)^2}{3}}}.$$
This unique $\|x\|$, by formula (\ref{e2}), defines unique fixed point of operator (\ref{e2o}) .

 3. In this example we take $a_{i1}=0$,  $i\geq 2$ and remaining $a_{ij}=1$. Then corresponding fixed point equation is

\begin{equation}\label{e3}
	x_{1}=\lambda_{1}, \ \
x_{i}=\lambda_{i}\left(\frac{1+\|x\|-\lambda_{1}}{1+\|x\|}\right)^2, \ \  i\geq 2.
\end{equation}
From this equation we get

$$\|x\|=\lambda_1+\frac{(1+\|x\|-\lambda_1)^2}{(1+\|x\|)^2}(\|\lambda\|-\lambda_1),$$ i.e.,
$$\|x\|^3+(2-\|\lambda\|)\|x\|^2+(2(1-\lambda_1)(1+\lambda_1-\|\lambda\|)-1)\|x\|-\lambda_1-(1-\lambda_1)^2(\|\lambda\|-\lambda_1)=0.$$

Similarly to the above mentioned examples one can show that this equation has unique solution:

%
%
%
%
%
%
%
%
$$\|x\|=\frac{1+\|\lambda\|}{3}-1+\alpha+\beta,$$
where
$$\alpha=\sqrt[3]{-\frac{B}{2}+\sqrt{(\frac{A}{3})^3+(\frac{B}{2})^2}},$$
$$\beta=\sqrt[3]{-\frac{B}{2}-\sqrt{(\frac{A}{3})^3+(\frac{B}{2})^2}},$$
$$A=\lambda_1(\|\lambda\|-\lambda_1)-
\frac{(1+\|\lambda\|)^2}{3},$$  $$B=-\frac{2(1+\|\lambda\|)^3}{27}+\frac{2\lambda_1(\|\lambda\|-\lambda_1)(1+\|\lambda\|)}{3}-\lambda_1^2(\|\lambda\|-\lambda_1).$$
Putting this unique $\|x\|$ in (\ref{e3}) we get the unique fixed point of the operator.

\section{An example for non-uniqueness}
In this section, for $k=2$, we consider an operator and show that it has more than one fixed points. Namely, we show that depending on parameters the operator has up to five fixed points.

Take $a_{11}=1$, $a_{1j}=1$, $a_{i1}=0$ for any $i\ne 1$, $j\ne 1$ and for other values of $i$, $j$ we take
\begin{equation}\label{GG}a_{ij}=\left\{\begin{array}{ll}
1, \ \ \mbox{if} \ \ i+j \ \ \mbox{is even}\\[2mm]
0, \ \ \mbox{if} \ \ i+j \ \ \mbox{is odd}.
\end{array}\right.
\end{equation}
Then the corresponding operator has the following form
\begin{equation}\label{e4o}
\begin{array}{lll}
	x'_{1}=\lambda_{1}\\[2mm]
x'_{2n}=\lambda_{2n}\left(\dfrac{1+\sum_{j=1}^{\infty}x_{2j}}{1+\|x\|}\right)^2\\[2mm]
x'_{2n+1}=\lambda_{2n+1}\left(\dfrac{1+\sum_{j=1}^{\infty}x_{2j+1}}{1+\|x\|}\right)^2.
\end{array}
\end{equation}
To find fixed points of this operator we introduce
$$M_1=\sum_{j=1}^{\infty}x_{2j+1}, \ \ M_2=\sum_{j=1}^{\infty}x_{2j},$$ $$L_1=\sum_{j=1}^{\infty}\lambda_{2j+1}, \ \ L_2=\sum_{j=1}^{\infty}\lambda_{2j}.$$
Note that $\|x\|=x_1+M_1+M_2$.

Then the fixed point equation of (\ref{e4o}) is reduced to
\begin{equation}\label{e4}
	\begin{array}{lll}
		x_{1}=\lambda_{1}\\[2mm]
		x_{2n}=\lambda_{2n}\left(\dfrac{1+M_2}{1+\|x\|}\right)^2\\[2mm]
		x_{2n+1}=\lambda_{2n+1}\left(\dfrac{1+M_1}{1+\|x\|}\right)^2.
	\end{array}
\end{equation}
 Summing the equations we get
\begin{equation}\label{ef4}
	\begin{array}{ll}
		M_1=L_1\left(\dfrac{1+M_1}{1+\lambda_1+M_1+M_2}\right)^2\\[2mm]
M_2=L_2\left(\dfrac{1+M_2}{1+\lambda_1+M_1+M_2}\right)^2.
\end{array}
\end{equation}
Thus each solution $(M_1, M_2)$ to (\ref{ef4}), by formula (\ref{e4}), uniquely defines a fixed point of operator (\ref{e4o}).

For simplicity we assume $L_1=L_2=L$.
Then from (\ref{ef4}) we get

$$M_1-M_2=L\left(\left(\dfrac{1+M_1}{1+\lambda_1+M_1+M_2}\right)^2-\left(\dfrac{1+M_2}{1+\lambda_1+M_1+M_2}\right)^2\right).$$
Consequently,
$$(M_1-M_2)(1+\lambda_1+M_1+M_2)^2=L\left(M_1-M_2\right)\left(2+M_1+M_2\right).$$
That is
\begin{equation}\label{MM}
	(M_1-M_2)\left[(1+\lambda_1+M_1+M_2)^2-L\left(2+M_1+M_2\right)\right]=0.
	\end{equation}
From (\ref{MM}) we get
$M_1=M_2$ or
\begin{equation}\label{M2}
(1+\lambda_1+M_1+M_2)^2=L\left(2+M_1+M_2\right).
\end{equation}

{\bf Case:} $M_1=M_2=M$. In this case we get
$$M=L\left(\dfrac{1+M}{1+\lambda_1+2M}\right)^2.$$
Denoting $a=\frac{4}{L}$ and $b=\frac{1+\lambda_1}{2}$ rewrite the last equation in the form
$$a=\frac{1}{M}\left(\frac{1+M}{b+M}\right)^2.$$
Introduce the following function  $$f(x)=\frac{1}{x}\left(\frac{1+x}{b+x}\right)^2.$$
We have
$$f'(x)=-\frac{(1+x)}{x^2(b+x)^3}(x^2+(3-b)x+b).$$
Note that if $b\leq 9$ then $f'(x)<0$ and the equation $a=f(x)$ has unique positive solution for each $a>0$. For $b>9$, from $f'(x)=0$ we get two positive solutions:
$$x_{1,2}=\frac{b-3\pm\sqrt{b^2-10b+9}}{2}.$$
Let $0<f(x_1)<f(x_2)$ then
$$\mbox{the number of positive solutions to} \ \ f(x)=a \ \  \mbox{is} \ \ \left\{\begin{array}{lll}
	1, \ \ \mbox{if} \ \ a\notin [f(x_1), f(x_2)]\\[2mm]
	2, \ \ \mbox{if} \ \ a\in\{f(x_1), f(x_2)\}\\[2mm]
	3, \ \ \mbox{if} \ \ a\in (f(x_1), f(x_2)).
\end{array}\right.
$$
Note that we have explicit form of  $f(x_1)$ and $f(x_2)$, but  they have bulky formula.
If $b=10$, for example, then in the case of 3 solutions, the above mentioned condition on $a$ becomes $\frac{1}{32}< a<\frac{4}{125}$. This condition for the initial parameters is as  $\lambda_1=19$, $125\leq L\leq 128$.

Recall that $a$ and $b$ depend on $L$ and $\lambda_1$. We denote
$$A_1=\{(L,\lambda_1)\in \mathbb R^2_+: a>0, b\leq 9\}\cup \{(L,\lambda_1)\in \mathbb R^2_+:a\notin [f(x_1), f(x_2)], b> 9\}.$$
$$A_2=\{(L,\lambda_1)\in \mathbb R^2_+:a\in \{f(x_1), f(x_2)\}, b> 9\}.$$
$$A_3=\{(L,\lambda_1)\in \mathbb R^2_+:a\in (f(x_1), f(x_2)), b> 9\}.$$
To give plots of these sets we rewrite above mentioned functions depending on initial parameters $L$ and $\lambda_1$:

$$x_1=\frac{\lambda_1-5-\sqrt{\lambda_1^2-18\lambda_1+17}}{4}, \ \ x_2=\frac{\lambda_1-5+\sqrt{\lambda_1^2-18\lambda_1+17}}{4},$$
$$f(x_1)=a \ \ \Leftrightarrow \ \  L=\frac{2\lambda_1^2+76\lambda_1-142+(2\lambda_1-34)\sqrt{\lambda_1^2-18\lambda_1+17}}{16},$$

$$f(x_2)=a\ \ \Leftrightarrow \ \  L=\frac{2\lambda_1^2+76\lambda_1-142-(2\lambda_1-34)\sqrt{\lambda_1^2-18\lambda_1+17}}{16}.$$
Using these equalities one draws the sets shown in Fig.\ref{A123}.
 \begin{figure}[h!]
	\includegraphics[width=9cm]{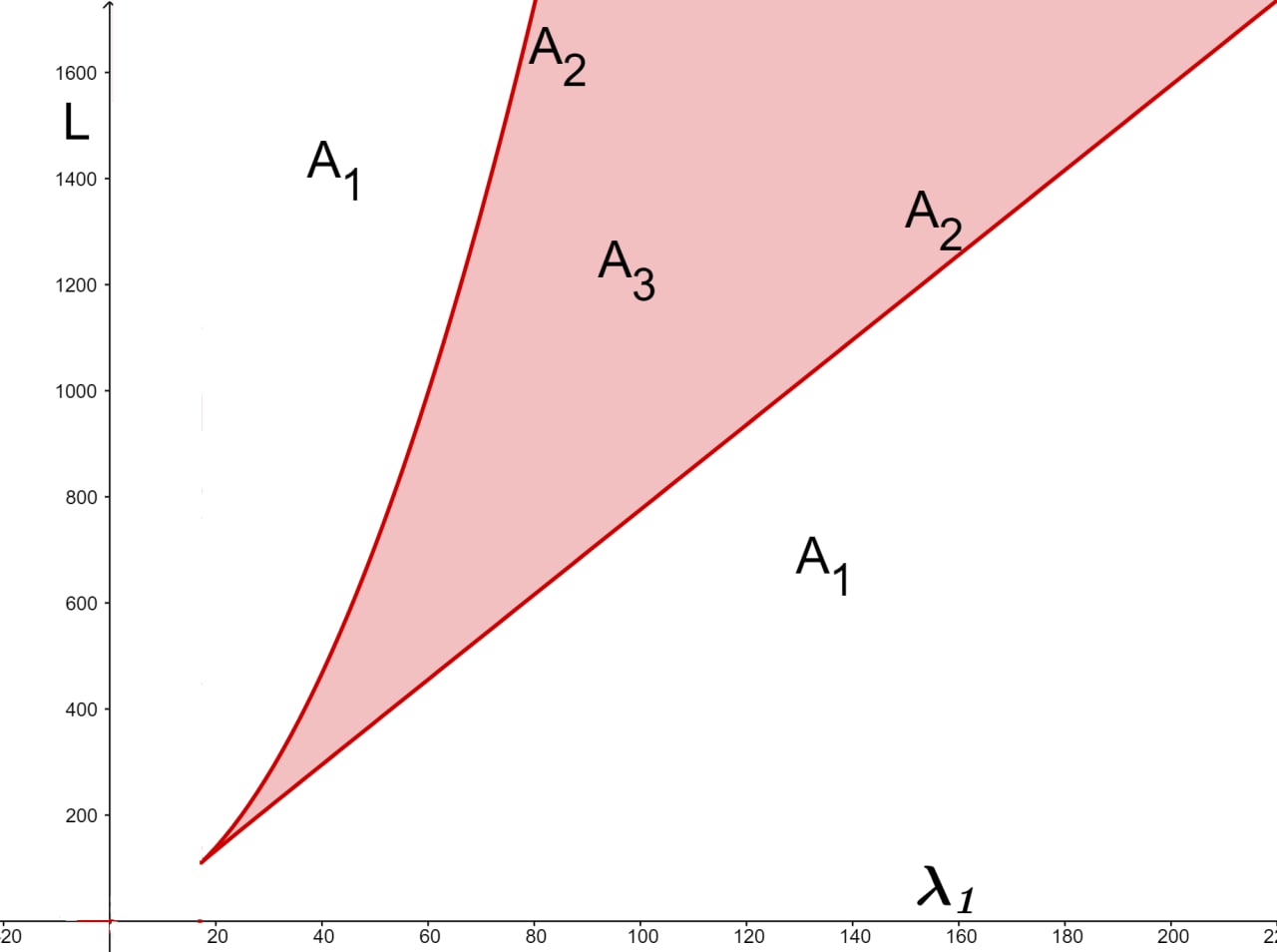}
	\caption{ The set $A_2$ is boundary of the red region, the set $A_3$ is inside of the region. The set $A_1$ is $\mathbb R^2\setminus (A_2\cup A_3)$. }\label{A123}
\end{figure}

{\bf Case}  (\ref{M2}): Now we assume $M_1\ne M_2$ and (\ref{M2}) holds. Denoting $M_1+M_2=t$ from (\ref{M2}) we get
\begin{equation}\label{Lt}
	L(2+t)=(1+\lambda_1+t)^2.
	\end{equation}
That is
$$t^2+(2+2\lambda_1-L)t+(1+\lambda_1)^2-2L=0.$$
It has solutions
$$t_{1,2}={1\over 2}(L-2-2\lambda_1\pm \sqrt{D}),$$
where $D=L(L+4-4\lambda_1)$. These solutions are defined iff $\lambda_1\leq 1+L/4$. Moreover, the condition $t_{1,2}>0$ gives that

1) $\lambda_1>\sqrt{2L}-1$,

2) $\lambda_1<\frac{L}{4}+1$,

3) $\lambda_1<\frac{L}{2}-1$.

These inequalities can be written as
\begin{equation}\label{L}
	\sqrt{2L}-1<\lambda_1<{L\over 4}+1, \ \ L>8.
	\end{equation}
Now for each $t_{1,2}=M_1+M_2$ from (\ref{ef4}) (recall $L_1=L_2=L$) we get
$$M_1=L\left(\dfrac{1+M_1}{1+\lambda_1+t_{1,2}}\right)^2, \ \ M_2=L\left(\dfrac{1+M_2}{1+\lambda_1+t_{1,2}}\right)^2.$$
Since $t_{1,2}$ satisfies (\ref{Lt}) the last system of equations can be written as
\begin{equation}\label{te}
	M_1=\dfrac{(1+M_1)^2}{2+t_{1,2}}, \ \	M_2=\dfrac{(1+M_2)^2}{2+t_{1,2}}.
	\end{equation}
From this system we get
$$\dfrac{(1+M_1)^2}{M_1}=\dfrac{(1+M_2)^2}{M_2},$$
that is satisfied only for $M_1=M_2$ and $M_1M_2=1$.
In the previous case we considered $M_1=M_2$, here remains $M_1M_2=1$.
From the first equation of (\ref{te}) we get
$$M_1={1\over 2}(t_{1,2}\pm \sqrt{t_{1,2}^2-4}).$$
This solution exists and positive iff $t_{1,2}\geq 2$.
Now, under condition (\ref{L}), we check  $t_{1,2}\geq 2$.

{\bf Sub-case:} $t_1\geq 2$. This inequality can be simplified to
\begin{equation}\label{se}
	\sqrt{L(L+4-4\lambda_1)}\geq 6+2\lambda_1-L.
	\end{equation}
{\bf Sub-sub-case:} $6+2\lambda_1-L\leq 0$. In this case (\ref{se}) is satisfied. Under condition (\ref{L}) we get
\begin{equation}
	8+4\sqrt{3}\leq L\leq 16, \ \ \ \sqrt{2L}-1< \lambda_1\leq {L\over 2}-3.
\end{equation}
{\bf Sub-sub-case:} $6+2\lambda_1-L> 0$. In this case the inequality (\ref{se}) is equivalent to
$$\lambda_1>{L\over 2}-3, \ \ \lambda_1^2+6\lambda_1-4L+9\leq 0.$$
It is easy to see that the last inequalities and (\ref{L}) have the following common solutions:
\begin{equation}\label{hi}
\max\left\{{L\over 2}-3,\,	\sqrt{2L}-1\right\}<\lambda_1\leq 2\sqrt{L}-3, \ \ L>2(1+\sqrt{2})^2.
\end{equation}
Denote
$$B_1=B\cup C,$$ where
$$B=\left\{(L,\lambda_1)\in \mathbb R^2_+: 8+4\sqrt{3}\leq L\leq 16, \ \ \sqrt{2L}-1<\lambda_1\leq {L\over 2}-3\right\},$$
$$C=\left\{(L,\lambda_1)\in \mathbb R^2_+: 	\max\left\{{L\over 2}-3,\,	\sqrt{2L}-1\right\}<\lambda_1\leq 2\sqrt{L}-3, \ \ L>2(1+\sqrt{2})^2\right\}.$$


{\bf Sub-case:} $t_2\geq 2$. This inequality can be simplified to
\begin{equation}\label{see}
	-\sqrt{L(L+4-4\lambda_1)}\geq 6+2\lambda_1-L.
\end{equation}
{\bf Sub-sub-case:} $6+2\lambda_1-L\leq 0$. In this case from (\ref{see}) we obtain
$$\lambda_1\leq {L\over 2}-3, \ \ \lambda_1^2+6\lambda_1-4L+9\geq 0.$$
These inequalities and (\ref{L}) then reduced to the following
\begin{equation}\label{tL}
	2\sqrt{L}-3\leq \lambda_1\leq \min\left\{{L\over 2}-3, {L\over 4}+1\right\}.
\end{equation}
{\bf Sub-sub-case:} $6+2\lambda_1-L> 0$. In this case the inequality (\ref{see}) has not any solution.

Denote
$$B_2=\left\{(L,\lambda_1)\in \mathbb R^2_+: L>8, 	2\sqrt{L}-3\leq \lambda_1\leq \min\left\{{L\over 2}-3, {L\over 4}+1 \right\}\right\}.$$

Note that for each solution $M_1$ with $M_1\ne M_2$ the value $M_2$ is uniquely determined by $M_2={1\over M_1}$.

 \begin{figure}[h!]
	\includegraphics[width=13.5cm]{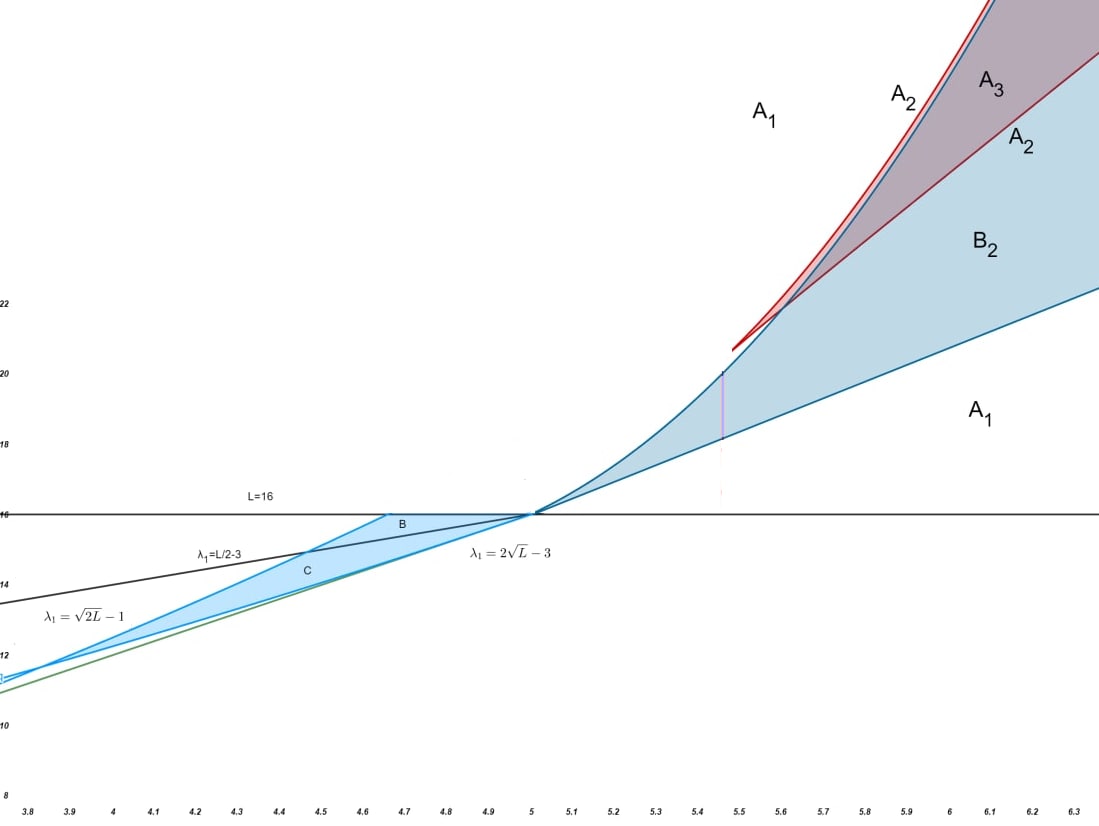}
	\caption{ The sets $A_i$, $i=1,2,3$ and $B$, $C$, $B_2$. }\label{B23}
\end{figure}
Now we summarize results of this section to the following (see Fig. \ref{B23})

\begin{thm}\label{tk} Let $\mathcal N$ be the number of fixed points of the operator (\ref{e4o}). Then
$$\mathcal N=\left\{\begin{array}{lllllllll}
	1, \ \ \mbox{if} \ \ (L,\lambda_1)\in A_1\setminus (B_1\cup B_2)\\[2mm]
	2, \ \ \mbox{if} \ \ (L,\lambda_1)\in A_2\setminus B_2\\[2mm]
	3,	\ \ \mbox{if} \ \ (L,\lambda_1)\in B_1\\[2mm]
	3,	\ \ \mbox{if} \ \ (L,\lambda_1)\in B_2\setminus (A_2\cup A_3)\\[2mm]
    3,	\ \ \mbox{if} \ \ (L,\lambda_1)\in A_3\setminus B_2\\[2mm]
    3, \ \ \mbox{if} \ \ (L,\lambda_1)\in A_1\cap B_1\\[2mm]
    3, \ \ \mbox{if} \ \ (L,\lambda_1)\in A_1\cap B_2\\[2mm]
    4, \ \ \mbox{if} \ \ (L,\lambda_1)\in A_2\cap B_2\\[2mm]
    5, \ \ \mbox{if} \ \ (L,\lambda_1)\in A_3\cap B_2
    \end{array}\right.$$
\end{thm}

\section{Application: Gibbs measures}
In this section we give an application of the above mentioned
results to construction of translation invariant Gibbs measures of spin systems defined on Cayley trees. 

{\bf Set-up.} Let us give the basic concepts for Gibbs
measures on a Cayley tree, and also fix some notation. 

A Cayley tree $\Gamma^{k}=(V, L)$  of order $k\geq 1$ is a graph without cycles and each its vertex has exactly $k+1$ edges. Here
 $V$ is the set of vertices of $\Gamma^{k}$ and $L$ is the set of its edges. If $l \in L$ then its endpoints $x, y \in V$ are called nearest neighbors and denoted by $l=\langle x, y\rangle$.

Let $d(x, y)$ be the distance between vertices $x$ and $y$ on the Cayley tree, i.e., number of edges of the shortest path connecting vertices $x$ and $y$.

For a fixed $x^{0} \in V$ we put
$$
W_{n}=\left\{x \in V \mid d\left(x, x^{0}\right)=n\right\}.
$$

 If $x \in W_{n}$ then the set $S(x)$ of direct successors of the vertex $x$ is
$$
S(x)=\left\{y_{i} \in W_{n+1} \mid d\left(x, y_{i}\right)=1, i=1,2, \ldots, k \right\}.
$$

For the HC-model with a countable number of states on the Cayley tree define configuration $\sigma=\{\sigma(x) \mid x \in V\}$ as a function from $V$ to the set of natural numbers $\mathbb{N}$.

Consider the set $\mathbb{N}$ as the set of vertices of some infinite graph $G$. Using the graph $G$. A  configuration $\sigma$ is called $G$-admissible on a Cayley tree if $\{\sigma(x), \sigma(y)\}$  is an edge of the graph $G$ for any nearest neighbors $x, y$ from $V$.

The set of $G$-admissible configurations is denoted by $\Omega^{G}$.

The activity set for the graph $G$ is the bounded function $\lambda: G \mapsto \mathbb{R}_{+}$ (where $\mathbb{R}_{+}$ is the set of positive real numbers).

Define the Hamiltonian of $G-$ HC-model  as
\begin{equation}\label{g1}
	H_{G}^{\lambda}(\sigma)= \begin{cases} \sum\limits_{x \in V} \ln \lambda_{\sigma(x)}, & \text { if } \sigma \in \Omega^{G}, \\ +\infty, & \text { if } \sigma \notin \Omega^{G}.\end{cases}
\end{equation}

The set of edges of the graph $G$ is denoted by  $L(G)$. Denote by $A \equiv A^{G}=\left(a_{i j}\right)_{i, j \in \mathbb N}$  the adjacency matrix of $G$, i.e.,

$$
a_{i j}=a_{i j}^{G}=\left\{\begin{array}{lll}
1 & \text { if } & \{i, j\} \in L(G), \\
0 & \text { if } & \{i, j\} \notin L(G).
\end{array}\right.
$$

\begin{defn}\label{Definition 3} (see \cite{Z1} and Chapter 12 of \cite{Ge}) A family of vectors $l=\left\{l_{xy}\right\}_{\langle x, y) \in L}$ with $l_ {xy}=$ $\left\{l_{xy}(i): i \in \mathbb{N}\right\} \in(0, \infty)^{\mathbb{N}}$ is called the boundary law for the Hamiltonian (\ref{g1}) if

1) for each $\langle x, y\rangle \in L$ there exists a constant $c_{x y}>0$ such that the consistency equation
\begin{equation}\label{5}
l_{xy}(i)=c_{xy}  \prod_{z \in \partial x \backslash\{y\}} \sum_{j \in \mathbb{N}} \lambda_i a_{ij}\lambda_{j} l_{zx}(j)
\end{equation}
holds for any $i \in \mathbb{N}$, where $\partial x$ is the set of nearest neighbors of $x$.

2) The boundary law $l$ is said to be normalisable if and only if
\begin{equation}\label{6}
\sum_{i \in \mathbb{N}}\left(\prod_{z \in \partial x} \sum_{j \in \mathbb{N}}\lambda_i a_{ij}\lambda_{j} l_{zx} (j)\right)<\infty
\end{equation}
for all $x \in V$.
\end{defn}
For given configuration $\omega$, graph $G$ with $A=(a_{ij})$, an edge $b=\langle x,y \rangle$, and $i=\omega(x)$, $j=\omega(y)$ define transfer matrices $Q_b$ by
\begin{equation}\label{Qd}
Q_b(i,j) =\lambda_ia_{ij}\lambda_j.
\end{equation}

Let $\omega_b=\{\omega(x), \omega(y)\},$ when $b=\langle x, y \rangle$.

For a finite subset $\Lambda\subset V$ define the (Markov) Gibbsian specification as
$$
\gamma_\Lambda(\sigma_\Lambda = \omega_\Lambda | \omega) = (Z_\Lambda)(\omega)^{-1} \prod_{b \cap \Lambda \neq \emptyset} Q_b(\omega_b).
$$

\begin{thm}\label{tz} \cite{Z1}
For any Gibbsian specification $\gamma$ with associated family of transfer matrices $(Q_b)_{b \in L}$  we have
\begin{enumerate}
\item Each {\it normalisable} boundary law $(l_{xy})_{x,y}$ for $(Q_b)_{b \in L}$ defines a unique Gibbs measure $\mu$ (corresponding to $\gamma$)  via the equation given for any connected set $\Lambda \subset V$
\begin{equation}\label{BoundMC}
	\mu(\sigma_{\Lambda \cup \partial \Lambda}=\omega_{\Lambda \cup \partial \Lambda}) = (Z_\Lambda)^{-1} \prod_{y \in \partial \Lambda} l_{y y_\Lambda}(\omega(y)) \prod_{b \cap \Lambda \neq \emptyset} Q_b(\omega_b),
\end{equation}
where for any $y \in \partial \Lambda$, $y_\Lambda$ denotes the unique nearest-neighbor of $y$ in $\Lambda$.
\item Conversely, every Gibbs measure $\mu$ admits a representation of the form (\ref{BoundMC}) in terms of a {\it normalisable} boundary law (unique up to a constant positive factor).
\end{enumerate}
\end{thm}
Denote $\mathbb{N}_1=\mathbb{N}\setminus\{1\}$, $\hat l_{zx} (j): =\lambda_{j} l_{zx} (j)$ and assume $\hat l_{x y}(1) \equiv 1$, then from (\ref{5}) (denoting $\lambda^k_i$ by $\lambda_i$) we obtain
\begin{equation}\label{7}
\hat l_{xy}(i)=\frac{\lambda_{i}}{\lambda_{1}} \prod_{z \in \partial x \backslash\{y\}} \frac{a_{i 1}+ \sum_{j \in \mathbb{N}_{1}} a_{ij} \hat l_{zx}(j)}{a_{11}+\sum_{j \in \mathbb{N}_{1}} a_{1 j}\hat l_{zx}(j)}.
\end{equation}

In this section we consider {\bf concrete graph} $G$ defined by the adjacency matrix (\ref{GG}), which we considered in the previous section.

Given a boundary law $\hat l_{xy}(i)$, we define $z_{i,x}=\hat l_{xy}(i)$ when $x$ is  direct successor of $y$, i.e. $x\in S(y)$, then (\ref{7})  can be written as (here without lost of generality we take $\lambda_1=1$).

\begin{equation}\label{8}
	z_{i,x}=\lambda_{i}\prod_{y \in S(x)}
	\frac{a_{i 1}+ \sum_{j \in \mathbb{N}_{1}} a_{ij} z_{j,y}}{a_{11}+\sum_{j \in \mathbb{N}_{1}} a_{1 j} z_{j,y}}.
\end{equation}
Thus the investigation of the Gibbs measures for Hamiltonian (\ref{g1}) for the graph $G$ given by matrix $A=(a_{ij})$ is reduced to finding solutions of (\ref{8}).

We give Gibbs measures corresponding to solutions mentioned in Theorem \ref{tk}. To do this we should first check normalisablity of solutions mentioned in this theorem. 

{\bf Normalisablity of solution (\ref{e4}).}

\begin{lemma}\label{new} If $\lambda\in l^{1}_{+}$ then any solution of the form (\ref{e4}) is normalisable.\end{lemma}

\begin{proof} Recall $Q(i,j):=\lambda_ia_{ij}\lambda_j$. The normalisablity of boundary laws can be reduced (see \cite{HRo}) to show that
	$$\sum_{i\in\mathbb{N}}\sum_{j\in\mathbb{N}}z_{x, i}Q(i,j)z_{y, j}<\infty, \ \ \forall \langle x, y\rangle \in L.$$
Now we check this condition for solution (\ref{e4}):
\begin{equation}\label{uf}	\sum_{i\in\mathbb{N}}\sum_{j\in\mathbb{N}}z_{x, i}Q(i,j)z_{y, j}=\sum_{i\in\mathbb{N}}\sum_{j\in\mathbb{N}}\lambda_i a_{ij}\lambda_{j}z_{x,i}z_{y,j}\leq \sum_{i\in\mathbb{N}}\lambda_i z_{x,i} \sum_{j\in\mathbb{N}}\lambda_{j}z_{y,j}.
	\end{equation}
	Since the solution is independent on the vertices of the Cayley tree,  for the RHS of (\ref{uf}) we have:
$$\sum_{i\in\mathbb{N}}\lambda_i z_{x,i} \sum_{j\in\mathbb{N}}\lambda_{j}z_{y,j}=\left(\sum_{i\in\mathbb{N}}\lambda_i z_{i}\right)^2.$$
	Therefore, the lemma follows from the following estimate (in the case of (\ref{e4})):
	
	$$\sum_{i\in\mathbb{N}}\lambda_i z_{i}=\lambda_1+
	\sum_{n\in \mathbb{N}}\lambda_{2n}^2\left(\dfrac{1+M_2}{1+\lambda_1+M_1+M_2}\right)^2$$ $$
	+\sum_{n\in \mathbb{N}}\lambda_{2n+1}^2\left(\dfrac{1+M_1}{1+\lambda_1+M_1+M_2}\right)^2
	< \lambda_1+\sum_{i\in\mathbb{N}_1} \lambda^{2}_i<+\infty.$$
Because, if $\lambda\in l^{1}_{+}$ then $\sum_{i\in\mathbb{N}_1}\lambda_i^2<\infty$.
\end{proof}

Now using Lemma \ref{new}, by Theorem \ref{tz} we conclude that each solution (\ref{e4}) defines a  (translation-invariant) Gibbs measure. Therefore as a corollary of Theorem \ref{tk} we get the following
\begin{thm}\label{to} 
Let $\mathcal{N}_G$ be the number of translation invariant Gibbs measures for Hamiltonian (\ref{g1}), corresponding to graph $G$ defined by (\ref{GG}), then  
$$\mathcal N_G=\left\{\begin{array}{lllllllll}
	1, \ \ \mbox{if} \ \ (L,\lambda_1)\in A_1\setminus (B_1\cup B_2)\\[2mm]
	2, \ \ \mbox{if} \ \ (L,\lambda_1)\in A_2\setminus B_2\\[2mm]
	3,	\ \ \mbox{if} \ \ (L,\lambda_1)\in B_1\\[2mm]
	3,	\ \ \mbox{if} \ \ (L,\lambda_1)\in B_2\setminus (A_2\cup A_3)\\[2mm]
	3,	\ \ \mbox{if} \ \ (L,\lambda_1)\in A_3\setminus B_2\\[2mm]
	3, \ \ \mbox{if} \ \ (L,\lambda_1)\in A_1\cap B_1\\[2mm]
	3, \ \ \mbox{if} \ \ (L,\lambda_1)\in A_1\cap B_2\\[2mm]
	4, \ \ \mbox{if} \ \ (L,\lambda_1)\in A_2\cap B_2\\[2mm]
	5, \ \ \mbox{if} \ \ (L,\lambda_1)\in A_3\cap B_2
\end{array}\right.$$
\end{thm}
\begin{rk} Lemma \ref{new} can be proved for the case of uniqueness of the fixed point, for operator $F$, mentioned in the previous sections, therefore the corresponding Hamiltonian (\ref{g1}) has unique translation invariant Gibbs measure.
\end{rk}


\section*{Statements and Declarations}

{\bf	Conflict of interest statement:}
On behalf of all authors, the corresponding author (U.A.Rozikov) states that there is no conflict of interest.

\section*{Data availability statements}
The datasets generated during and/or analysed during the current study are available from the corresponding author on reasonable request.

\end{document}